\documentclass[final,1p]{elsarticle}
\usepackage{etoolbox}
\makeatletter
\patchcmd{\ps@pprintTitle}{\footnotesize\itshape
       Preprint submitted to \ifx\@journal\@empty Elsevier
       \else\@journal\fi\hfill\today}{\relax}{}{}
\makeatother

\usepackage{pslatex,bm} 
\usepackage{amsfonts,color,morefloats}
\usepackage{amssymb,amsmath,latexsym,amsthm}

\newcommand{\wt}{{\mathrm{wt}}}

\newcommand{\tr}{{\mathrm{Tr}}}
\newcommand{\gf}{{\mathrm{GF}}}

\newcommand{\C}{{\mathcal{C}}}

\newcommand{\cR}{{\mathcal{R}}}

\newcommand{\bc}{{\mathbf{c}}}

\newtheorem{theorem}{Theorem}

\newtheorem{open}{Research Problem}

\begin{document}

\begin{frontmatter}

\title{The construction and weight distributions of all projective binary linear codes
}


\author{Cunsheng Ding}
 \ead{cding@ust.hk}

 \address{Department of Computer Science and Engineering,
The Hong Kong University of Science and Technology, Clear Water Bay, Kowloon, Hong Kong}

\date{}

\begin{abstract} 
Boolean functions can be used to construct binary linear codes in many ways, and vice versa.  
The objective of this short article is to point out a connection between the weight distributions of all 
projective binary linear codes and the Walsh spectra of all Boolean functions. New research problems 
are also proposed. 
\end{abstract}


\end{frontmatter}

\section{Introduction}\label{sec-intro} 

Let $q$ be a prime and let $r=q^m$ for some positive integer $m$. 
An $[n,\, k,\, d]$ code $\C$ over $\gf(q)$ is a $k$-dimensional subspace of $\gf(q)^n$ with minimum 
(Hamming) distance $d$. A linear code $\C$ is called \emph{projective} if its dual code has minimum 
distance at least $3$. 
Let $A_i$ denote the number of codewords with Hamming weight $i$ in a code
$\C$ of length $n$. The {\em weight enumerator} of $\C$ is defined by
$
1+A_1z+A_2z^2+ \cdots + A_nz^n.
$ 
The sequence $(1, A_1, A_2, \cdots, A_n)$ is called the \emph{weight distribution} of the code $\C$. 
A code $\C$ is said to be a $t$-weight code  if the number of nonzero
$A_i$ in the sequence $(A_1, A_2, \cdots, A_n)$ is equal to $t$.

Boolean functions are functions from $\gf(2^m)$ or $\gf(2)^m$ to $\gf(2)$, where $m$ is a positive integer. They are important building blocks 
for certain types of stream ciphers, and can also be employed to construct binary codes in many ways. 
Conversely, binary linear codes can be used to construct Boolean functions in different ways. The objective 
of this article is to point out a connection between the weight distributions of all projective binary linear codes and the Walsh spectra of all Boolean functions. New research problems are also proposed. 

\section{Mathematical foundations} 

\subsection{Group characters in $\gf(q)$}

An {\em additive character} of $\gf(q)$ is a nonzero function $\chi$ 
from $\gf(q)$ to the set of nonzero complex numbers such that 
$\chi(x+y)=\chi(x) \chi(y)$ for any pair $(x, y) \in \gf(q)^2$. 
For each $b\in \gf(q)$, the function
\begin{eqnarray}\label{dfn-add}
\chi_b(c)=\epsilon_p^{\tr(bc)} \ \ \mbox{ for all }
c\in\gf(q) 
\end{eqnarray}
defines an additive character of $\gf(q)$, where and whereafter $\epsilon_p=e^{2\pi \sqrt{-1}/p}$ is 
a primitive complex $p$th root of unity and $\tr$ is the absolute trace function. When $b=0$,
$\chi_0(c)=1 \mbox{ for all } c\in\gf(q), 
$ 
and is called the {\em trivial additive character} of
$\gf(q)$. The character $\chi_1$ in (\ref{dfn-add}) is called the
{\em canonical additive character} of $\gf(q)$. 
It is known that every additive character of $\gf(q)$ can be 
written as $\chi_b(x)=\chi_1(bx)$ \cite[Theorem 5.7]{LN97}.

\subsection{Boolean functions and their expressions} 

A function $f$ from $\gf(2^m)$ or $\gf(2)^m$ to $\gf(2)$ is called a {\em Boolean function}. 
A function $f$ from $\gf(2^m)$ to $\gf(2)$ is called {\em linear} if $f(x+y)=f(x)+f(y)$ for all $(x, y) \in \gf(2^m)^2$. 
A function $f$ from $\gf(2^m)$ to $\gf(2)$ is called {\em affine} if $f$ or $f-1$ is linear. 

The {\em Walsh transform} of $f: \gf(2^m) \to \gf(2)$ is defined by 
\begin{eqnarray}\label{eqn-WalshTransform2}
\hat{f}(w)=\sum_{x \in \gf(2^m)} (-1)^{f(x)+\tr(wx)} 
\end{eqnarray} 
where $w \in \gf(2^m)$. The {\em Walsh spectrum} of $f$ is the following multiset 
$$ 
\left\{\left\{ \hat{f}(w): w \in \gf(2^m) \right\}\right\}. 
$$ 

Let $f$ be a Boolean function from $\gf(2^m)$ to $\gf(2)$. The \emph{support} of $f$ is defined to be 
\begin{eqnarray}\label{eqn-Booleanfsupport}
D_f=\{x \in\gf(2^m) : f(x)=1\} \subseteq \gf(2^m). 
\end{eqnarray}
Clearly, $f \mapsto D_f$ is a one-to-one correspondence between the set of Boolean functions from  
$\gf(2^m)$ to $\gf(2)$  and the power set of $\gf(2^m)$.

\section{A fundamental construction of linear codes}\label{sec-2ndgconst}

Throughout this section, let $q$ be a prime power and let $r=q^m$, where $m$ is a positive integer. 
Let $\tr$ denote the trace function from $\gf(r)$ to $\gf(q)$ unless otherwise stated. 

Let $D=\{d_1, \,d_2, \,\ldots, \,d_n\} \subseteq \gf(r)$. 
We define a code of 
length $n$ over $\gf(q)$ by 
\begin{eqnarray}\label{eqn-maincode191} 
\C_{D}=\{(\tr(xd_1), \tr(xd_2), \ldots, \tr(xd_n)): x \in \gf(r)\},   
\end{eqnarray}  
and call $D$ the \emph{defining set} of this code $\C_{D}$. 
Since the trace function is linear, the code $\C_D$ is linear.  By definition, the 
dimension of the code $\C_D$ is at most $m$. 

Different orderings of the elements of $D$ give different linear codes $\C_{D}$, which are however 
permutation equivalent. Hence, we do not distinguish these codes obtained 
by different orderings of the elements in $D$. It should be noticed that the defining set $D$ could be a multiset, i.e., some elements in $D$ may be 
the same.

Define for each $x \in \gf(r)$, 
\begin{eqnarray*}\label{eqn-mcodeword}
\bc_{x}=(\tr(xd_1), \,\tr(xd_2), \,\ldots, \,\tr(xd_n)).    
\end{eqnarray*} 
The Hamming weight $\wt(\bc_x)$ of $\bc_x$ is $n-N_x(0)$, where  
$$ 
N_x(0)=\left|\{1 \le i \le n: \tr(xd_i)=0\}\right| 
$$ 
for each $x \in \gf(r)$. 

It is easily seen that for any $D=\{d_1,\,d_2,\,\ldots, \,d_n\} \subseteq \gf(r)$
we have 
\begin{eqnarray*}\label{eqn-hn3}  
qN_x(0) 
= \sum_{i=1}^n \sum_{y \in \gf(q)} \tilde{\chi}_1 (y\tr(xd_i)) \nonumber 
%&=& \sum_{i=1}^n \sum_{y \in \gf(q)} \chi_1(yxd_i) \nonumber  \\ 
%&=& n + \sum_{i=1}^n \sum_{y \in \gf(q)^*} \chi_1(yxd_i)  \nonumber \\ 
= n + \sum_{y \in \gf(q)^*} \chi_1(yxD),  
\end{eqnarray*} 
where $\chi_1$ and $\tilde{\chi}_1$ are the canonical additive characters of $\gf(r)$ and $\gf(q)$, respectively, 
$aD$ denotes the set 
$\{ad: d \in D\}$, and $\chi_1(S):=\sum_{x \in S} \chi_1(x)$ for any subset $S$ of $\gf(r)$.  
Hence, 
\begin{eqnarray}\label{eqn-weight}
\wt(\bc_x)=n-N_x(0)=\frac{(q-1)n-\sum_{y \in \gf(q)^*} \chi_1(yxD)}{q}. 
\end{eqnarray}
Thus, the computation of the weight distribution of the code $\C_D$ reduces to the determination of the value distribution 
of the character sum 
$$ 
\sum_{y \in \gf(q)^*} \sum_{i=1}^n \chi_1(yxd_i). 
$$

This construction technique was employed many years ago  for obtaining linear codes with a few weights 
(see, for example, \cite{Wolfmann75},  \cite{DN07}, \cite{DLN} and \cite{Dingconf}), and is called the 
defining-set construction of linear codes. Recently, this trace construction of linear codes has attracted 
a lot of attention, and a huge amount of linear codes with good parameters have been obtained. The 
ollowing theorem shows that the trace construction is fundamental. 

\begin{theorem}\label{thm-defsetcode}
Every $[n, k]$ code over $\gf(q)$ can be expressed as $\C_D$ for some defining set 
$D \subseteq \gf(q^k)$.      
\end{theorem} 

\begin{proof}
Let $(g_{1j}, g_{2j}, \ldots, g_{kj})^T$ denote the $j$th column of a generator matrix 
of the code for $1 \leq j \leq n$. Define 
$$ 
f_j(x)=(x_1, x_2, \ldots, x_k)(g_{1j}, g_{2j}, \ldots, g_{kj})^T,  
$$ 
where $x=(x_1, x_2, \ldots, x_k) \in \gf(q)^k$. By definition, the code is the set  
$$ 
\{(f_1(x), f_2(x), \ldots, f_n(x)): x \in \gf(q)^k\}. 
$$ 

Let $\{\alpha_1, \alpha_2, \ldots, \alpha_k\}$ be a basis of $\gf(q^k)$ over $\gf(q)$, 
and let $\{\beta_1, \beta_2, \ldots, \beta_k\}$ denote its dual basis. For each $j$ 
with $1 \leq j \leq n$, define 
\begin{eqnarray}\label{eqn-oct9mor}
d_j=\sum_{i=1}^k g_{ij} \beta_i
\end{eqnarray} 
and $D=\{d_1, d_2, \ldots, d_n\} \subseteq \gf(q^k)$. 
For $x=(x_1, x_2, \ldots, x_k) \in \gf(q)^k$, define 
$$ 
x'=\sum_{i=1}^k x_i \alpha_i \in \gf(q^k). 
$$ 
Clearly, we have 
$$ 
\tr_{q^k/q}(d_j x')=\sum_{i=1}^k x_i g_{ij}=f_j(x). 
$$
Consequently, 
\begin{eqnarray*}
\lefteqn{\{(f_1(x), \ldots, f_n(x)): x \in \gf(q)^k\}} \\ 
&& =\{\tr_{q^k/q}(d_1x'), \ldots, \tr_{q^k/q}(d_nx'): x' \in \gf(q^k)\} \\
&& =\C_{D}. 
\end{eqnarray*} 
This completes the proof. 
\end{proof}

Theorem \ref{thm-defsetcode} is a direct consequence of the classical result that every linear function from 
$\gf(q^m)$ to $\gf(q)$ can be expressed as $\tr(ax)$ for some $a \in \gf(q^m)$ \cite{LN97}. The proof of 
Theorem \ref{thm-defsetcode}  clearly shows that the defining-set construction is equivalent to the generator 
matrix construction of all linear codes. Hence, it is impossible to find out the first one who introduced the 
defining-set construction.  This construction technique was employed many years ago  for obtaining linear 
codes with a few weights (see, for example, \cite{Wolfmann75},  \cite{DN07}, \cite{DLN} and \cite{Dingconf}). 
The 
weight formula in (\ref{eqn-hn3}) tells us that an advantage of the defining-set approach over the generator-matrix 
approach is that the former 
can make full use of results about character sums for determining the parameters and weight distributions  
of linear codes. This advantage has been demonstrated in a lot of recent references on the defining-set 
construction of linear codes.

A slightly different version of Theorem  \ref{thm-defsetcode} was proved in \cite{Xiang16}.  Theorem  \ref{thm-defsetcode} 
and its proof above are 
refined ones in \cite{HWW20}.

\section{The construction and weight distributions of all projective binary linear codes}\label{sec-supportcodes}

Let $f$ be a function from $\gf(2^m)$ to $\gf(2)$, and let $D_f$ be the support of $f$ defined in (\ref{eqn-Booleanfsupport}).  Let $n_f=|D_f|$.  The following theorem was proved in \cite{Ding15}. 

\begin{theorem}\label{thm-BooleanCodes} 
Let $f$ be a function from $\gf(2^m)$ to $\gf(2)$, and let $D_f$ be the support of $f$. If $2n_f + \hat{f}(w) \neq 0$ 
for all $w \in \gf(2^m)^*$, then $\C_{D_f}$ is a binary linear code with 
length $n_f$ and dimension $m$, and its weight distribution is given by the following multiset: 
\begin{eqnarray}\label{eqn-WTBcodes}
\left\{\left\{ \frac{2n_f+\hat{f}(w)}{4}: w \in \gf(2^m)^*\right\}\right\} \cup \left\{\left\{ 0 \right\}\right\}. 
\end{eqnarray} 
\end{theorem}

Theorem \ref{thm-BooleanCodes} establishes a connection between the set of Boolean functions $f$ such that 
$2n_f + \hat{f}(w) \neq 0$ for all $w \in \gf(2^m)^*$ and a class of binary linear codes. 
The determination of the weight distribution of the binary linear code $\C_{D_f}$ is equivalent to that of the Walsh 
spectrum of the Boolean function $f$ satisfying $2n_f + \hat{f}(w) \neq 0$ for all $w \in \gf(2^m)^*$. 
When the Boolean function $f$ is selected properly, the code $\C_{D_f}$ has 
only a few weights and may have good parameters. A lot of binary linear codes $\C_{D_f}$ with a few weights were 
reported in \cite{Ding16}.   

Theorem \ref{thm-BooleanCodes} was generalized into the following in \cite{Ding16}. 

\begin{theorem}\label{thm-BooleanCodesG}
Let $f$ be a function from $\gf(2^m)$ to $\gf(2)$, and let $D_f$ be the support of $f$. Let $e_w$ denote the 
multiplicity of the element $\frac{2n_f+\hat{f}(w)}{4}$ and $e$ the multiplicity of 0  in the multiset 
of (\ref{eqn-WTBcodes}). 
Then $\C_{D_f}$ is a binary linear code with length $n_f$ and dimension $m-\log_2 e$, and 
the weight distribution of the code is given by  
\begin{eqnarray*}
\frac{2n_f+\hat{f}(w)}{4} \mbox{ with frequency } \frac{e_w}{e} 
\end{eqnarray*}
for all $\frac{2n_f+\hat{f}(w)}{4}$ in the multiset of (\ref{eqn-WTBcodes}). 
\end{theorem}

Theorem \ref{thm-BooleanCodesG} says that every Boolean function can be used to construct a binary linear code 
with the defining-set construction whose weight distribution is determined by the Walsh spectrum of the Boolean 
function. Conversely, we have the following.  

\begin{theorem}\label{thm-oct16morning}
Let $\C$ be any projective binary linear code. Then there is a Boolean function $f$ such that $\C=\C_{D_f}$. Furthermore, 
Let $f$ be a function from $\gf(2^m)$ to $\gf(2)$, and let $D_f$ be the support of $f$. Let $e_w$ denote the 
multiplicity of the element $\frac{2n_f+\hat{f}(w)}{4}$ and $e$ the multiplicity of 0  in the multiset 
of (\ref{eqn-WTBcodes}). 
Then $\C_{D_f}$ is a binary linear code with length $n_f$ and dimension $m-\log_2 e$, and 
the weight distribution of the code is given by  
\begin{eqnarray*}
\frac{2n_f+\hat{f}(w)}{4} \mbox{ with frequency } \frac{e_w}{e} 
\end{eqnarray*}
for all $\frac{2n_f+\hat{f}(w)}{4}$ in the multiset of (\ref{eqn-WTBcodes}).   
\end{theorem}

\begin{proof}
By Theorem \ref{thm-defsetcode}, there is a set $D \subset \gf(2^m)$ for some positive integer $m$ such that 
$\C=\C_{D}$. Since $\C$ is projective, $D$ does not contain repeated elements. We now define a Boolean function 
$f$ from $\gf(2^m)$ to $\gf(2)$ as 
\begin{eqnarray}\label{eqn-oct16morning}
f(x) = \left\{ 
\begin{array}{ll}
1 & \mbox{ if } x \in D, \\
0 & \mbox{ otherwise.} 
\end{array}
\right. 
\end{eqnarray}  
By definition, the support $D_f$ of $f$ is $D$. It then follows that 
$$ 
\C=\C_{D}=\C_{D_f}. 
$$ 
The desired conclusion on the weight distribution of $\C$ then follows from Theorem \ref{thm-BooleanCodesG}. 
\end{proof} 

Note that Theorem \ref{thm-oct16morning} is a direct consequence of Theorems \ref{thm-defsetcode} 
and  \ref{thm-BooleanCodesG}. It gives a general approach to the computation of the weight distribution 
of projective binary linear codes $\C_D$. The procedure is the following. 
\begin{itemize}
\item Let $\C_D$ be a projective binary linear code constructed with the defining-set approach, where $D \subset \gf(2^m)$ 
for some positive $m$ and $D$ does not contain repeated elements. The first step is to construct the characteristic 
Boolean function $f$ of $D$, which was defined in (\ref{eqn-oct16morning}). 
\item The second step is to compute the Walsh spectrum of the Boolean function $f$.   
\end{itemize}  
Hence, determining the weight distribution of any projective binary linear code is equivalent to 
determining the Walsh spectrum of the corresponding Boolean function. 

The proofs of Theorems \ref{thm-defsetcode} and \ref{thm-oct16morning} clearly show that every projective binary 
linear code $\C$ with dimension $k$ gives a Boolean function $f_{\C}$ from $\gf(2^k)$ to $\gf(2)$ whose Walsh 
spectrum is completely determined by the weight distribution of $\C$. The Boolean function $f_{\C}(x)$ is constructed 
as follows: 
\begin{itemize}
\item Select a generator matrix of $\C$ and a basis of $\gf(2^k)$ over $\gf(2)$ (see the proof of Theorem \ref{thm-defsetcode}). 
\item Construct the set $D=\{d_1, d_2, \ldots, d_n\}$, where $d_j$ was defined in 
 (\ref{eqn-oct9mor}). 
\item Construct $f_{\C}$ as the characteristic function of $D$, i.e., $f_{\C}(x)=1$ if and only if $x \in D$.          
\end{itemize}  
This $f_{\C}(x)$ depends on the choice of the generator matrix and a basis of $\gf(2^k)$ over $\gf(2)$. 
Hence, such a code $\C$ gives many Boolean functions with the same Walsh spectrum. Below we propose 
some research problems in this direction. 

\begin{open} 
Let $\C$ denote the binary Golay code with parameters $[23, 12, 7]$ or its extended code.   Study the Boolean functions $f_{\C}(x)$ 
from $\gf(2^{12})$ to $\gf(2)$. 
\end{open} 

\begin{open} 
Let $\C$ denote the binary MacDonald code with parameters $[2^k-2, k, 2^{k-1}-1]$ (i.e., a punctured 
code of the binary Simplex code with parameters $[2^k-1, k, 2^{k-1}]$. Then $\C$ has only two nonzero weights $2^{k-1}$ 
and $2^{k-1}-1$. Study the Boolean functions $f_{\C}(x)$. 
\end{open} 

\begin{open} 
Let $\cR_2(\ell, m)$ denote the binary Reed-Muller code of order $1 \leq \ell<m$. Study the Boolean functions 
$f_{\cR_2(\ell, m)}(x)$. 
\end{open} 

\begin{open} 
Let $m \geq 3$ and let $C$ denote the binary Hamming code with parameters $[2^m-1, 2^m-1-m, 3]$. 
Study the Boolean functions $f_{\C}(x)$. 
\end{open} 

\begin{open} 
Let $\C$ denote a binary irreducible cyclic code.  Study the Boolean functions $f_{\C}(x)$.
\end{open} 

\begin{open} 
Let $\C$ denote a binary BCH code. Study the Boolean functions $f_{\C}(x)$. 
\end{open} 

\begin{open} 
For binary quadratic residue codes $\C$, study the Boolean functions $f_{\C}(x)$.
\end{open} 

\begin{open} 
Let $m \geq 4$ be even and let $\C$ denote the binary linear code with parameters $[2^m, m+2, 2^{m-1}-2^{(m-2)/2}]$ 
in Theorem 14.4 in \cite[p. 336]{Dingbk18}, study the Boolean functions $f_{\C}(x)$ from $\gf(2^{m+2})$ to 
$\gf(2)$. 
\end{open} 

\begin{open} 
Let $m \geq 4$ be even and let $\C$ denote the binary linear code with parameters $[2^{m-1}-2^{(m-2)/2}, \, m+1, \,  2^{m-2}-2^{(m-2)/2}]$ 
in Theorem 14.9 in \cite[p. 341]{Dingbk18}, study the Boolean functions $f_{\C}(x)$ from $\gf(2^{m+1})$ to 
$\gf(2)$. 
\end{open} 

There are a lot of binary linear codes $\C$ with only two nonzero weights documented in \cite{CK85}. 
The corresponding Boolean functions $f_{\C}$ should be very interesting.  It would be very interesting to study 
the Boolean functions $f_{\C}(x)$ for binary three-weight codes $\C$. There are many three-weight binary codes 
in the literature. Little work in this direction is done. The reader is cordially invited to join the venture in this 
direction. 

\section{A special case of the defining-set construction of linear codes} 

Let notation be the same as in Section \ref{sec-2ndgconst}. In this section, we consider a special case of 
the defining-set construction in \eqref{eqn-maincode191}.  

Assume that $m=2h$. Let $\{u_1, u_2\}$ be a basis of $\gf(q^m)$ over $\gf(q^h)$, and let $\{v_1, v_2\}$ 
be its dual basis. Then 
$$ 
d_i=d_{i,1}v_1+d_{i,2}v_2 
$$
where $d_{i,j} \in \gf(q^h)$. Similarly each $x \in \gf(q^m)$ can be expressed as 
$$ 
x=x_1u_1+x_2u_2, 
$$
where $x_i \in \gf(q^h)$. It then follows that 
$$ 
\tr_{q^m/q}(xd_i)= \tr_{q^h/q}(\tr_{q^m/q^h}(xd_i))= \tr_{q^h/q}(d_{i,1}x_1+d_{i,2}x_2). 
$$
Consequently, the code in \eqref{eqn-maincode191} can be expressed as 
\begin{eqnarray}\label{eqn-oct72020}
\C_{D} &=& \{ ( \tr_{q^h/q}(d_{i,1}x_1+d_{i,2}x_2) )_{i=1}^n: (x_1, x_2) \in \gf(q^h) \times \gf(q^h)  \}  \nonumber \\
            &=&  \{ ( \tr_{q^h/q}(e_1x_1+e_2x_2) )_{(e_1, e_2) \in E}: (x_1, x_2) \in \gf(q^h) \times \gf(q^h)  \},   
\end{eqnarray}
where $E=\{(d_{1,1}, d_{1,2}),  (d_{2,1}, d_{2,2}), \ldots, (d_{n,1}, d_{n,2})\}$. Thus, the construction of \eqref{eqn-oct72020} 
is a special case of the general defining-set construction, and was studied in some recent papers. 

Let $q=2$ and consider the code $\C_E$ in \eqref{eqn-oct72020}. Let $f$ be the characteristic Boolean function of $E$. 
Then $f$ is a function from $\gf(2^h) \times \gf(2^h)$ to $\gf(2)$. Then the weight distribution of the binary code $\C_E$ in \eqref{eqn-oct72020} is given in a similar way as in Theorem \ref{thm-oct16morning}.

\section{Concluding remarks} 

The Boolean function construction of projective binary linear codes $\C_{D_f}$ gives a coding-theoretical characterisation 
of bent and other special Boolean functions with Theorems  \ref{thm-BooleanCodesG} and 
\ref{thm-oct16morning}. For example, $f$ is bent if and only if the code $\C_{D_f}$ has the weight distribution in 
Table II in \cite{Ding15}. In addition, other connections between projective binary codes and Boolean functions could 
also be developed.

\end{document}